\documentclass[sigconf]{acmart}

\usepackage{algorithm}
\usepackage{algorithmic}
\usepackage{balance}
\usepackage{multirow}
\usepackage{amsmath}
\usepackage{mathtools}
\usepackage{amsthm}

\usepackage{booktabs}
\usepackage{xspace}
\newcommand{\ours}[0]{FedHUR\xspace}

\AtBeginDocument{%
  }

\copyrightyear{2026}
\acmYear{2026}
\setcopyright{cc}
\setcctype{by}
\acmConference[CIKM '26]{Proceedings of the 35th ACM International Conference on Information and Knowledge Management}{November 07--11, 2026}{Rome, Italy}
\acmBooktitle{Proceedings of the 35th ACM International Conference on Information and Knowledge Management (CIKM '26), November 07--11, 2026, Rome, Italy}
\acmDOI{10.1145/3799682.3841027}
\acmISBN{979-8-4007-2539-5/2026/11}

\begin{document}

\title{FedHUR: Learning Hierarchical Utility-Guided Client Relations for Personalized Federated Recommendation}

\author{Mingzhe Han}
\orcid{0000-0002-4911-6093}
\affiliation{
  \institution{Fudan University}
  \city{Shanghai}
  \country{China}
}
\email{mzhan22@m.fudan.edu.cn}

\author{Jiahao Liu}
\orcid{0000-0002-5654-5902}
\affiliation{%
  \institution{Fudan University}
  \city{Shanghai}
  \country{China}
}
\email{jiahaoliu21@m.fudan.edu.cn}

\author{Dongsheng Li}
\orcid{0000-0003-3103-8442}
\affiliation{
  \institution{Microsoft Research Asia}
  \city{Shanghai}
  \country{China}
}
\email{dongshengli@fudan.edu.cn}

\author{Jiankui Zhou}
\orcid{0009-0003-6986-1300}
\affiliation{%
  \institution{Fudan University}
  \city{Shanghai}
  \country{China}
}
\email{jkzhou24@m.fudan.edu.cn}

\author{Hansu Gu}
\orcid{0000-0002-1426-3210}
\affiliation{
  \institution{Independent}
  \city{Seattle}
  \country{United States}
}
\email{hansug@acm.org}

\author{Peng Zhang}
\orcid{0000-0002-9109-4625}
\authornote{Corresponding author.}
\affiliation{
  \institution{Fudan University}
  \city{Shanghai}
  \country{China}
}
\email{zhangpeng\_@fudan.edu.cn}

\author{Ning Gu}
\orcid{0000-0002-2915-974X}
\affiliation{
  \institution{Fudan University}
  \city{Shanghai}
  \country{China}
}
\email{ninggu@fudan.edu.cn}

\author{Tun Lu}
\orcid{0000-0002-6633-4826}
\authornotemark[1]
\affiliation{
  \institution{Fudan University}
  \city{Shanghai}
  \country{China}
}
\email{lutun@fudan.edu.cn}

\renewcommand{\shortauthors}{Mingzhe Han et al.}

\begin{abstract}
Federated recommendation enables collaborative model training while keeping user interaction data on local clients.
A central problem in federated recommendation is how to aggregate useful information across clients for personalized recommendation.
Existing personalized aggregation methods usually construct client relations from predefined parameter-based assumptions, such as parameter similarity or complementarity, and use these relations to determine aggregation weights.
However, such methods construct a single global relation, which is insufficient to capture the hierarchical and multi-granularity nature of user relations in recommendation.
Moreover, these predefined relations cannot directly reflect whether the related clients can improve prediction performance after aggregation. 
To address these limitations, we propose FedHUR, a federated recommendation framework for learning hierarchical utility-guided client relations.
FedHUR takes item-item filters as the object for relation construction and aggregation.
Specifically, it first aggregates and clusters each client's local information to obtain global hierarchical information.
Each client computes hierarchical utility signals based on its local information and the global hierarchical information, indicating which collaborative information is useful for improving its prediction.
The server uses these utility signals to retrieve clients that are useful to that client for further personalized aggregation.
Extensive experiments on five real-world datasets show that FedHUR consistently outperforms existing federated recommendation baselines, demonstrating the effectiveness of hierarchical utility-guided client relation learning.
Code is available at https://github.com/Mingzhe-Han/FedHUR.
\end{abstract}



\maketitle

\section{Introduction}
Recommendation systems~\cite{liu2022parameter,liu2026distribution} learn user preferences from user interaction data, such as clicks, ratings, and purchases~\cite{jing2023contrastive,zhao2025dual,yang2026drsorec}. 
However, these interactions may reveal private user information, and regulations such as the General Data Protection Regulation (GDPR)~\cite{voigt2017eu} have imposed stricter requirements on the collection of personal data~\cite{gu2026inter,zhao2025social,liu2023recommendation,liu2025filtering,liu2026hidden}. 
Federated recommendation~\cite{mcmahan2017communication,kairouz2021advances} has therefore become a practical training paradigm. It keeps user interaction data on local clients and updates the server model using trained client models rather than collecting raw interaction data. 
In this way, the server recommendation model can benefit from distributed data while protecting private user information.

A central problem in federated recommendation is how to aggregate information across clients~\cite{li2023revisiting,shi2025fedawa,tang2024adapted}.
Traditional federated methods~\cite{mcmahan2017communication,li2020federated,karimireddy2020scaffold} usually aggregate client parameters by weighted averaging, where the aggregation weights are mainly determined by local data sizes.
These weights determine how much each client contributes to the aggregated model.
However, in recommendation systems, clients usually have different data distributions, so different clients may provide different levels of useful information for a given \textbf{target client}.
Therefore, traditional weighted averaging may fail to aggregate the most useful information for each individual client, leading to suboptimal recommendation performance.
To address this problem, existing methods construct personalized client relations~\cite{ye2023personalized,zhang2024gpfedrec} for each client and aggregate parameters based on the weights calculated by these relations.
These aggregation methods help the target client pay more attention to related clients, thereby improving the recommendation performance.

Despite their effectiveness, existing client relations are still limited in terms of granularity and utility.
\textbf{In terms of granularity, a single global relation may not fully model the complex client relations.}
Existing methods usually construct a single global relation to describe the correlation between clients.
However, user relations in recommendation systems are naturally hierarchical and multi-granularity~\cite{zang2023contrastive,qi2021hierec,wang2022target}.
For example, a candidate user may be helpful to the target user in sports items but conflict with the target user in movie items, which requires us to use multiple weights to model the relations between these two users~\cite{liu2025agentcf++,zhoudisentangling}.
Therefore, it is necessary to construct \textbf{hierarchical} client relations.
\textbf{In terms of utility, related clients may not be truly helpful for prediction.}
Existing methods usually measure relations between clients by calculating predefined parameter correlation (e.g., similarity or complementarity) between clients, and then use these relations to construct aggregation weights.
However, these methods only rely on assumptions that correlated clients can provide useful information.
Such relations are directly calculated from client parameters before aggregation, while ignoring the performance after aggregation.
We believe that the goal of aggregation is to aggregate useful information, which refers to the information that can improve the performance after aggregation.
Therefore, it is necessary to construct client relations \textbf{guided by post-aggregation utility}.

These limitations indicate that personalized federated recommendation requires \textbf{hierarchical utility-guided client relations}. 
However, obtaining such relations is non-trivial. 
First, \textbf{hierarchical relation spaces are difficult to define.} 
In federated recommendation, client relations can only be inferred from communicable item information.
However, item information is usually sparse and noisy, since many items are interacted with by only a few clients and local item parameters can be biased by limited user behaviors.
If we define relations over the whole item space, useful relations may be overlooked because they only appear in specific item groups.
If we define independent relations for isolated item groups or individual items, the learned relations may become unreliable, redundant, or even conflicting.
Second, \textbf{reliable aggregation utility is difficult to compute. }
Intuitively, the utility should be obtained by the recommendation performance after aggregation. 
This means we need to evaluate the prediction result for every candidate client, leading to heavy computation and communication costs. 

Inspired by graph signal processing~\cite{liu2023personalized,xia2024hierarchical,huang2017collaborative}, these challenges can be alleviated from the perspective of item graph filtering.
\textbf{For hierarchical relation modeling,} the graph filter view allows us to organize item relations in a coarse-to-fine manner.
Coarse relations capture stable common information over large item groups, while fine relations further describe relations within each group.
Like a tree structure, fine relations inherit information from their coarse parent groups, so client retrieval can consider hierarchical relations while avoiding noisy and redundant relation estimation.
\textbf{For utility computation,} graph signal processing provides an efficient way to analyze which items still need information from other clients.
Instead of repeatedly performing parameter aggregation and validation, it can directly analyze the client's filtering result on the item graph and derive a utility signal.
This signal indicates which part of the item space requires external collaborative information.
In this way, the aggregation utility can be estimated efficiently without additional aggregation and validation.

Based on this idea, we propose \ours, a federated recommendation framework for \textbf{learning hierarchical utility-guided client relations}.
Following graph-based federated recommendation methods~\cite{han2025fedcia}, \ours uses item-item filters as the aggregation target.
In each communication round, each client first trains a local recommendation model and uploads item embeddings to the server.
The server constructs filters from the uploaded embeddings and learns hierarchical utility-guided client relations through two key modules.
The first module is the \textbf{hierarchical relation mining module}.
It first clusters items into coarse item groups and computes coarse filters for each client.
Based on them, \ours further divides each coarse item group into finer item groups and computes fine filters.
In this way, \ours obtains filters at different granularities, which are used to construct hierarchical client relations in the following retrieval module.
The second module is the \textbf{utility-guided client retrieval module}.
Based on the filters at different granularities, each client computes its utility signals and uploads them to the server.
The server then retrieves the candidate clients and trains a scorer to estimate the usefulness of a candidate client for the target client.
The estimated scores are normalized as aggregation weights, which are exactly the hierarchical utility-guided client relations learned by \ours.
Finally, the server aggregates candidate filters according to these weights and sends the personalized filter back to each client for prediction.

Our main contributions are summarized as follows:

\begin{itemize}
    \item We propose to mine \textbf{hierarchical utility-guided client relations} for personalized federated recommendation. 
    The proposed relations are guided by aggregation utility and organized hierarchically, addressing the limitations of existing methods: selected clients may not be truly helpful, and a single global relation cannot capture hierarchical client relations.

    \item We propose \textbf{\ours}, which contains a \textbf{hierarchical relation mining module} and a \textbf{utility-guided client retrieval module}. The former constructs hierarchical relations by grouping items at different granularities, while the latter computes utility scores to estimate which candidate clients are useful for each target client.

    \item Extensive experiments on real-world datasets show that \ours consistently improves recommendation performance over existing federated recommendation baselines, demonstrating the effectiveness of hierarchical utility-guided client relations.
\end{itemize}

\section{Related Work}

\subsection{Federated Recommendation}

Federated learning enables multiple clients to collaboratively train a model without uploading raw local data~\cite{mcmahan2017communication,wang2024horizontal,sun2024survey}. 
It has been widely introduced into recommendation to protect private user interactions. 
In federated recommendation, user interaction data are stored on local clients, while the server coordinates model training through exchanged model information~\cite{zhang2023lightfr,zhang2026transfr}.
Early federated recommendation methods mainly adapt conventional recommendation models to the federated setting. 
For example, FedRec trains recommendation models locally and communicates model updates to protect user interactions~\cite{lin2020fedrec}. 
FedMF applies matrix factorization in federated recommendation and considers privacy protection for uploaded gradients~\cite{chai2020secure}. 
FedNCF further extends neural collaborative filtering to the federated setting~\cite{perifanis2022federated}.

Although these methods enable privacy-preserving recommendation, they usually rely on simple aggregation strategies, such as weighted averaging according to local data sizes. 
However, user behaviors in recommendation are heterogeneous, so different clients may provide different useful information for a target client. 
Therefore, simply aggregating all client updates into a single global model may be suboptimal for each individual client. 

\subsection{Personalized Aggregation in Federated Recommendation}

Personalized federated recommendation~\cite{hu2025p2fedrec} aims to improve the recommendation performance of each individual client by aggregating information in a client-specific manner~\cite{arivazhagan2019federated,t2020personalized,li2021ditto}. 
Some methods estimate client relations from item parameters, assuming that similar clients should share more information~\cite{zhang2024gpfedrec}. 
Some studies further design more specific aggregation strategies. 
FedCA argues that similarity alone is insufficient and introduces complementary clients for aggregation~\cite{zhang2026beyond}. 
FedRAP decomposes item embeddings into shared and personalized components to improve client personalization~\cite{li2024federated}. 
These methods demonstrate that aggregation based on relations is important for federated recommendation.

Despite their effectiveness, most of them construct client relations from predefined priors, such as similarity, complementarity, or collaborative correlation, and usually describe client relations with a single global relation.
Unlike these methods, \ours learns hierarchical utility-guided client relations to retrieve useful candidate information for each target client at different levels.

\section{Preliminaries}

\subsection{Recommendation Algorithms}

Let $\mathcal{U}$ be the set of users and $\mathcal{I}$ be the set of items.
For each user $u \in \mathcal{U}$, the recommendation task is to estimate the preference score of user $u$ on an item $i \in \mathcal{I}$ from observed interaction data.
We denote the interaction dataset by $\mathcal{D}=\{(u,i,y_{ui})\},$ where $y_{ui}\in\{0,1\}$ indicates whether the interaction between user $u$ and item $i$ is observed.

A general recommendation model can be formulated as $\hat{y}_{ui}=R(u,i;\theta),$ where $\hat{y}_{ui}$ denotes the predicted score and $\theta$ denotes all trainable parameters.
In embedding-based models, users and items are first mapped into latent representations.
The prediction function can be written as
\[
\hat{y}_{ui}
=
P\big(\mathbf{e}_u, \mathbf{e}_i; \theta_p\big),
\quad
\mathbf{e}_u=E_u(u;\theta_u),\quad
\mathbf{e}_i=E_i(i;\theta_i),
\]
where $E_u(\cdot)$ and $E_i(\cdot)$ are the user and item embedding functions, and $P(\cdot)$ is a prediction function such as inner product or an MLP.

\subsection{Federated Recommendation Algorithms}

In federated recommendation, interaction data are kept on distributed clients rather than collected by the server.
Assume there are $K$ clients, and client $k$ owns a private local dataset $\mathcal{D}_k$.
Each client updates its model using its local interactions.

Different from general federated learning, recommendation models usually contain both user parameters and item parameters.
User parameters are closely related to private user preferences, so they are kept on local clients and are not uploaded to the server.
Therefore, federated recommendation methods usually upload and aggregate item or public model information.

The local model of client $k$ can be divided into local user parameters $\theta_{u,k}$ and uploaded parameters $\phi_k$, where $\phi_k$ may include item embeddings, item parameters, or public prediction parameters.
At communication round $t$, each client trains its local model on $\mathcal{D}_k$ and uploads $\phi_k^{t+1}$ to the server.
A traditional aggregation strategy is weighted averaging:
\[
\phi^{t+1}
=
\sum_{k=1}^{K}
\omega_k \phi_k^{t+1},
\quad
\omega_k=
\frac{|\mathcal{D}_k|}
{\sum_{j=1}^{K}|\mathcal{D}_j|}.
\]
The aggregated parameters $\phi^{t+1}$ are then distributed to clients for the next communication round, while the user parameters remain local.

\subsection{Relation-Based Personalized Aggregation}

Recent personalized federated recommendation methods assign aggregation weights according to client relations.
For a target client $k$ and a candidate client $j$, a relation score is usually computed from their uploaded information or client representations as $r_{k,j} = \mathrm{sim}(\mathbf{z}_k,\mathbf{z}_j),$ where $\mathbf{z}_k$ denotes the information used to describe client $k$, and $\mathrm{sim}(\cdot,\cdot)$ is a relation function.

The relation scores are then normalized into aggregation weights and used to aggregate the uploaded information:
\[
\begin{aligned}
\alpha_{k,j}
=
\frac{\exp(r_{k,j})}
{\sum_{m=1}^{K}\exp(r_{k,m})},
\tilde{\phi}_k
=
\sum_{j=1}^{K}\alpha_{k,j}\phi_j.
\end{aligned}
\]
where $\phi_j$ denotes the uploaded item or public information from client $j$, such as item embeddings, item parameters, or other collaborative representations.
The personalized aggregated information $\tilde{\phi}_k$ is then used by the target client together with its local user parameters.
\begin{figure*}[t]
\centering
\includegraphics[width=2.0\columnwidth]{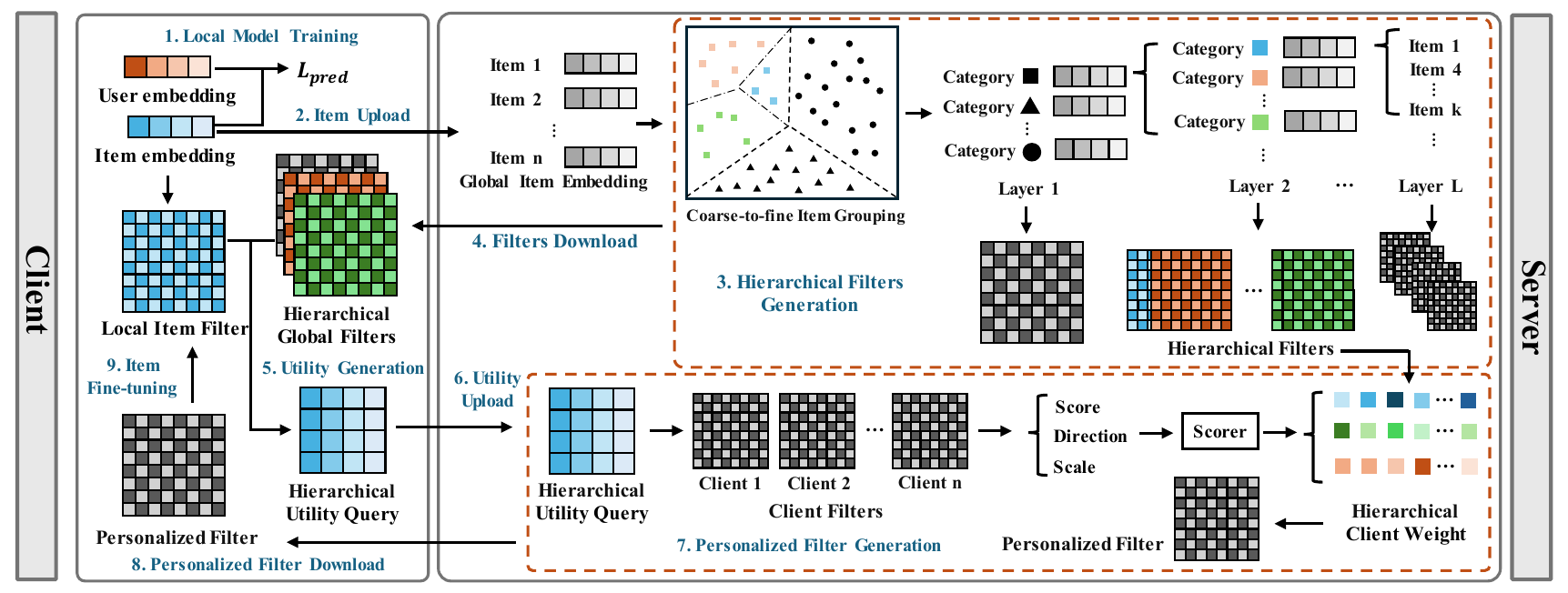}
\caption{
The illustration of the \ours framework. We only illustrate a single client for simplicity.
}
\label{fig:fra}
\end{figure*}

\section{Method}
\subsection{Overview}

In this section, we introduce \ours, a federated recommendation framework for learning hierarchical utility-guided client relations.
Figure~\ref{fig:fra} illustrates one communication round of \ours.

In each communication round, clients first train local item embeddings from their private interactions and upload the item embeddings to the server.
Given the uploaded item embeddings, the server computes the global item embeddings and global item-item filter.
The server then partitions the items according to the global embeddings and computes hierarchical item-item filters based on the hierarchical relation mining module.
These filters are distributed to clients.
Each client computes hierarchical utility signals locally and uploads these signals to the server.
Given the uploaded utility signals, the server trains a scorer to generate the aggregation weights for each client based on the utility-guided client retrieval module.
Finally, the server aggregates useful candidate filters across hierarchical client relations to obtain personalized filters, which are sent back to clients for local fine-tuning.

\subsection{Local Training and Upload}

At the beginning of each round, each client $k$ trains the local recommendation model on its private data $\mathcal{D}_k$:
\begin{equation}
\min_{\theta_{u,k},\phi_k}
\mathcal{L}^{(k)}_{rec}
=
\frac{1}{|\mathcal{D}_k|}
\sum_{(u,i,y_{ui})\in\mathcal{D}_k}
\ell_{rec}
\left(
R(u,i;\theta_{u,k},\phi_k),
y_{ui}
\right),
\label{eq:local_training}
\end{equation}
where $\theta_{u,k}$ denotes the user parameters kept by client $k$, and $\phi_k$ denotes the item parameters that can be communicated.
In this work, $\phi_k$ is instantiated as the local item embeddings $\mathbf{E}_k$.
Therefore, the user parameters $\theta_{u,k}$ remain on local clients and are never uploaded.
Client $k$ only uploads the item embedding matrix $\mathbf{E}_k$ to the server.

\subsection{Global Information Aggregation}

Before learning hierarchical utility-guided client relations, \ours first constructs global information.
After receiving item embeddings from all clients, the server computes the global item embedding and the global item-item filter:
\begin{equation}
\mathbf{E}_g
=
\frac{1}{K}
\sum_{k=1}^{K}
\mathbf{E}_k,
\qquad
\mathbf{S}_g
=
\frac{1}{K}
\sum_{k=1}^{K}
\mathbf{E}_k\mathbf{E}_k^\top
\in
\mathbb{R}^{|\mathcal{I}|\times|\mathcal{I}|}.
\label{eq:global_information}
\end{equation}
Here, $\mathbf{E}_g$ is used to construct hierarchical item groups, while $\mathbf{S}_g$ serves as the global base filter for later personalized aggregation.

\subsection{Hierarchical Relation Mining Module}

Existing federated recommendation methods usually construct a single global relation between clients.
However, as discussed in the introduction, user relations in recommendation systems are hierarchical and multi-granularity.
Two clients may be related at a coarse level, while their relations may become different at a finer level.
Therefore, \ours introduces a hierarchical relation mining module to construct hierarchical relations from coarse to fine.

\ours takes filters as the aggregation objects. Accordingly, the hierarchical relation mining module should construct filters at different granularities for subsequent client relation calculation.
At a coarse level, \ours groups items into large item groups and computes coarse-level filters.
Then, within each coarse group, \ours further divides items into finer groups and computes finer filters.
The finer filter inherits the corresponding coarse filter value, so that relation modeling can preserve coarse-level common information while adding finer-level details.
By repeating this process, \ours obtains multi-level filters that support later utility-guided relation computation.

\subsubsection{Coarse-to-Fine Item Grouping}

Starting from the global item embedding $\mathbf{E}_g$, the server first partitions the whole item set $\mathcal{I}$ into coarse item groups.
In this work, we use k-means clustering as the default partition method for its simplicity and efficiency, while other clustering methods can also be used.
Then, each coarse group is further divided into finer item groups by applying k-means within the group.
This process is repeated for $L$ levels, where $L$ is a hyperparameter.

We denote the item groups at level $\ell$ as:
\begin{equation}
\mathcal{G}^{(\ell)}
=
\left\{
\mathcal{G}^{(\ell)}_1,
\mathcal{G}^{(\ell)}_2,
\dots,
\mathcal{G}^{(\ell)}_{M_\ell}
\right\},
\qquad
\ell=1,\dots,L,
\end{equation}
where $M_\ell = |\mathcal{G}^{(\ell)}|$ denotes the number of item groups at level $\ell$, and a smaller $\ell$ indicates a coarser level.
The first level $\mathcal{G}^{(1)}$ contains coarse item groups obtained by applying k-means to all item embeddings in $\mathbf{E}_g$.
For $\ell>1$, the groups at level $\ell$ are obtained by further applying k-means within each group from level $\ell-1$.
Thus, the item space is decomposed into coarse-to-fine item groups.

\subsubsection{Coarse-Level Relation Filters}

For each coarse group $\mathcal{G}^{(1)}_a$, client $k$ computes its group embedding by averaging item embeddings in the group:
\begin{equation}
\mathbf{z}^{(1)}_{k,a}
=
\frac{1}{|\mathcal{G}^{(1)}_a|}
\sum_{i\in \mathcal{G}^{(1)}_a}
\mathbf{e}_{k,i}.
\end{equation}
The coarse-level filter of client $k$ is then computed by the inner product between group embeddings:
\begin{equation}
\mathbf{F}^{(1)}_k[a,b]
=
\left(\mathbf{z}^{(1)}_{k,a}\right)^\top
\mathbf{z}^{(1)}_{k,b},
\qquad
a,b=1,\dots,M_1.
\end{equation}
The global coarse-level filter is obtained by averaging the local filters:
\begin{equation}
\mathbf{F}^{(1)}_g
=
\frac{1}{K}
\sum_{k=1}^{K}
\mathbf{F}^{(1)}_k.
\end{equation}

\subsubsection{Finer-Level Relation Filters}

After constructing the coarse-level filters, \ours further builds finer relation filters within each coarse group.
At level $\ell>1$, each group $\mathcal{G}^{(\ell-1)}_a$ from the previous level is further divided into several smaller groups.
For each smaller group $\mathcal{G}^{(\ell)}_b \subset \mathcal{G}^{(\ell-1)}_a$, client $k$ computes its group embedding by averaging the item embeddings inside the group:
\begin{equation}
\mathbf{z}^{(\ell)}_{k,b}
=
\frac{1}{|\mathcal{G}^{(\ell)}_b|}
\sum_{i\in \mathcal{G}^{(\ell)}_b}
\mathbf{e}_{k,i}.
\end{equation}

The finer-level relation filter is then computed between smaller groups.
To preserve the hierarchical relation, each finer-level relation inherits the corresponding parent-level relation value:
\begin{equation}
\mathbf{F}^{(\ell)}_{k,a}[b,c]
=
\beta_{\ell-1}
\mathbf{F}^{(\ell-1)}_{g}[a,a]
+
\left(
\mathbf{z}^{(\ell)}_{k,b}
\right)^\top
\mathbf{z}^{(\ell)}_{k,c},
\qquad
\mathcal{G}^{(\ell)}_b,\mathcal{G}^{(\ell)}_c
\subset
\mathcal{G}^{(\ell-1)}_a .
\label{eq:finer_filter}
\end{equation}
where $\beta_{\ell-1}$ controls the contribution of the inherited parent-level relation.

The global finer-level filter is obtained by averaging over clients:
\begin{equation}
\mathbf{F}^{(\ell)}_{g,a}
=
\frac{1}{K}
\sum_{k=1}^{K}
\mathbf{F}^{(\ell)}_{k,a}.
\end{equation}

By repeating this construction from coarse to fine, \ours obtains relation filters at multiple granularities.
These filters preserve coarse-level relations while adding finer-level information, and are used for later utility-guided client relation computation.

\subsection{Utility-Guided Client Retrieval Module}

After obtaining relation filters at different levels, \ours needs to identify which candidate clients are useful for the target client at each level.
Existing federated recommendation methods usually compute client relations by similarity or complementarity between model parameters.
However, these relations do not directly validate whether the candidate client's information can improve the target client's prediction after aggregation.
A straightforward solution is to aggregate each candidate client's information to the target client and validate the prediction improvement.
However, this requires repeated aggregation and validation, which is computationally expensive and difficult to scale.

To address this problem, \ours formulates utility computation as a retrieval-based estimation task.
Rather than repeatedly computing the utility improvement after aggregation, \ours estimates the information required by the target client and uses such information as a query to retrieve helpful clients.
Specifically, for each target client, we construct a query based on its local and global filters.
This query represents the filter direction most required by the target client. 
The server uses this query to score all candidate clients and find the clients that are most helpful to the target client.
Therefore, the expensive process of utilities validation is transformed into a query-based retrieval process, where \ours directly searches for clients that can provide the information required by the target client.

For clarity, we describe the retrieval process at an arbitrary level $\ell$ and omit the item-group index.
The same procedure is applied to all relation filters constructed by the hierarchical relation mining module.

\subsubsection{Utility Query Construction}

For target client $k$, we use the global filter $\mathbf{F}^{(\ell)}_g$ as the base filter at level $\ell$.
Client $k$ first computes a local filter $\mathbf{H}^{(\ell)}_k$ from its local item information at the current level.
This local filter summarizes the item information observed by client $k$ and reflects the local item structure.

Instead of directly evaluating the performance of each candidate after aggregation, \ours estimates the utility query as:
\begin{equation}
\mathbf{U}^{(\ell)}_k
=
\mathbf{H}^{(\ell)}_k
-
\frac{1}{2}
\left(
\mathbf{H}^{(\ell)}_k
\mathbf{F}^{(\ell)}_g
+
\mathbf{F}^{(\ell)}_g
\mathbf{H}^{(\ell)}_k
\right).
\label{eq:utility_query}
\end{equation}
Intuitively, $\mathbf{U}^{(\ell)}_k$ describes the filter direction that the target client requires to improve its local prediction.
A candidate client is useful if its residual filter can provide information aligned with this query.

\begin{theorem}[Validity of Utility Query]
Assume that $\mathbf{H}^{(\ell)}_k$ and $\mathbf{F}^{(\ell)}_g$ are symmetric filters.
The utility query
\begin{equation}
\mathbf{U}^{(\ell)}_k
=
\mathbf{H}^{(\ell)}_k
-
\frac{1}{2}
\left(
\mathbf{H}^{(\ell)}_k
\mathbf{F}^{(\ell)}_g
+
\mathbf{F}^{(\ell)}_g
\mathbf{H}^{(\ell)}_k
\right)
\end{equation}
provides the utility direction required by target client $k$.
\end{theorem}

\begin{proof}
To justify the direction of $\mathbf{U}^{(\ell)}_k$, we introduce a local embedding reconstruction objective.
Let $\mathbf{Z}^{(\ell)}_k$ denote the item/group representation matrix of client $k$ at level $\ell$, and let
$
\mathbf{H}^{(\ell)}_k
=
\mathbf{Z}^{(\ell)}_k
\mathbf{Z}^{(\ell)\top}_k .
$

From the perspective of graph signal filtering, the item-item filter
$F$ can be used to propagate and reconstruct the local item signal
of client $k$. Specifically, given the local item/group representation
$Z_k^{(\ell)}$, the filtered prediction is
$\widehat{Z}_k^{(\ell)\top}=Z_k^{(\ell)\top}F .$
Therefore, the residual
\[
Z_k^{(\ell)\top}-Z_k^{(\ell)\top}F
\]
measures the part of the local collaborative signal that cannot be
well reconstructed by the current filter. To identify the filter update
direction that reduces this local prediction error, we introduce the
following reconstruction objective:
\[
J_k^{(\ell)}(F)
=
\frac{1}{2}
\left\|
Z_k^{(\ell)\top}
-
Z_k^{(\ell)\top}F
\right\|_F^2 .
\]

The objective can be rewritten as:
\begin{align}
\mathcal{J}^{(\ell)}_k(\mathbf{F})
&=
\frac{1}{2}
\left\|
\mathbf{Z}^{(\ell)\top}_k
(\mathbf{I}-\mathbf{F})
\right\|_F^2 \notag \\
&=
\frac{1}{2}
\operatorname{Tr}
\left(
(\mathbf{I}-\mathbf{F})^\top
\mathbf{Z}^{(\ell)}_k
\mathbf{Z}^{(\ell)\top}_k
(\mathbf{I}-\mathbf{F})
\right) \notag \\
&=
\frac{1}{2}
\operatorname{Tr}
\left(
(\mathbf{I}-\mathbf{F})^\top
\mathbf{H}^{(\ell)}_k
(\mathbf{I}-\mathbf{F})
\right).
\end{align}

Expanding the trace form gives:
\begin{align}
\mathcal{J}^{(\ell)}_k(\mathbf{F})
=
\frac{1}{2}
\operatorname{Tr}
\left(
\mathbf{H}^{(\ell)}_k
\right)
-
\operatorname{Tr}
\left(
\mathbf{H}^{(\ell)}_k
\mathbf{F}
\right)
+
\frac{1}{2}
\operatorname{Tr}
\left(
\mathbf{F}^{\top}
\mathbf{H}^{(\ell)}_k
\mathbf{F}
\right).
\end{align}

Taking the derivative with respect to $\mathbf{F}$, we obtain:
\begin{equation}
\nabla_{\mathbf{F}}
\mathcal{J}^{(\ell)}_k(\mathbf{F})
=
\mathbf{H}^{(\ell)}_k
\mathbf{F}
-
\mathbf{H}^{(\ell)}_k .
\end{equation}
Therefore, at the global filter $\mathbf{F}^{(\ell)}_g$, the negative gradient is:
\begin{equation}
-
\nabla_{\mathbf{F}}
\mathcal{J}^{(\ell)}_k(\mathbf{F}^{(\ell)}_g)
=
\mathbf{H}^{(\ell)}_k
-
\mathbf{H}^{(\ell)}_k
\mathbf{F}^{(\ell)}_g .
\end{equation}

Since the item-item filter is symmetric, we take the symmetric part of this negative gradient:
\begin{align}
\operatorname{Sym}
\left(
\mathbf{H}^{(\ell)}_k
-
\mathbf{H}^{(\ell)}_k
\mathbf{F}^{(\ell)}_g
\right)
&=
\frac{1}{2}
\left[
\mathbf{H}^{(\ell)}_k
-
\mathbf{H}^{(\ell)}_k
\mathbf{F}^{(\ell)}_g
+
\left(
\mathbf{H}^{(\ell)}_k
-
\mathbf{H}^{(\ell)}_k
\mathbf{F}^{(\ell)}_g
\right)^\top
\right] \notag \\
&=
\mathbf{H}^{(\ell)}_k
-
\frac{1}{2}
\left(
\mathbf{H}^{(\ell)}_k
\mathbf{F}^{(\ell)}_g
+
\mathbf{F}^{(\ell)}_g
\mathbf{H}^{(\ell)}_k
\right).
\end{align}

This is exactly the utility query $\mathbf{U}^{(\ell)}_k$.
Thus, $\mathbf{U}^{(\ell)}_k$ provides the utility direction used to evaluate whether candidate filters can help target client $k$.
\end{proof}

Directly uploading $\mathbf{U}^{(\ell)}_k$ can be expensive for large filters.
Thus, client $k$ uploads a projected query:
\begin{equation}
\mathbf{Q}^{(\ell)}_k
=
\operatorname{Norm}_{q}
\left(
\mathbf{U}^{(\ell)}_k
\mathbf{P}^{(\ell)}
\right),
\label{eq:projected_query}
\end{equation}
where $\mathbf{P}^{(\ell)}$ is a shared random projection matrix, and $\operatorname{Norm}_{q}(\cdot)$ normalizes each projected query direction to unit $\ell_2$ norm.
The projected query keeps the main matching direction of $\mathbf{U}^{(\ell)}_k$ while reducing the upload cost.

\subsubsection{Utility-Guided Retrieval}

After constructing the utility query, the server uses it to retrieve useful candidate clients for the target client.
At level $\ell$, let $\mathbf{Z}^{(\ell)}_j$ denote the item or group embedding matrix of candidate client $j$.
Since client $k$ uploads the projected query $\mathbf{Q}^{(\ell)}_k$, the server computes the retrieval score by projecting the candidate representation onto the utility directions of the target client:
\begin{equation}
s^{(\ell)}_{k,j}
=
\left\|
\mathbf{Q}^{(\ell)\top}_{k}
\mathbf{Z}^{(\ell)}_{j}
\right\|_F^2 .
\label{eq:projected_retrieval_score}
\end{equation}
Here, $\mathbf{Q}^{(\ell)\top}_{k}\mathbf{Z}^{(\ell)}_{j}$ measures how much the candidate representation aligns with the utility directions required by target client $k$.
A larger score indicates that candidate client $j$ can provide more information along the directions currently needed by the target client.

\subsubsection{Learned Retrieval Scorer}

The retrieval score in Eq.~\eqref{eq:projected_retrieval_score} provides a direct estimate of information utility.
However, this score only gives a total matching value.
It does not distinguish whether this matching is concentrated on important directions or spread over noisy directions, and it may also be affected by the scale of the item set.
Therefore, after obtaining the retrieval score, \ours uses a lightweight scorer to refine the score with more detailed matching information.

Specifically, for target client $k$ and candidate client $j$ at level $\ell$, we decompose the direct retrieval score into direction-level matching strengths.
Let $\mathbf{q}^{(\ell)}_{k,t}$ denote the $t$-th column of $\mathbf{Q}^{(\ell)}_k$.
The scorer input is defined as:
\begin{equation}
\begin{aligned}
\mathbf{x}^{(\ell)}_{k,j}
&=
\operatorname{Norm}_{j}
\left(
\left[
\boldsymbol{\rho}^{(\ell)}_{k,j},
\nu^{(\ell)}_j,
\sqrt{s^{(\ell)}_{k,j}}
\right]
\right), \\
\rho^{(\ell)}_{k,j,t}
&=
\left\|
\mathbf{q}^{(\ell)\top}_{k,t}
\mathbf{Z}^{(\ell)}_{j}
\right\|_2,
\qquad
\nu^{(\ell)}_j
=
\frac{\|\mathbf{Z}^{(\ell)}_j\|_F}{\sqrt{m_\ell d}} .
\end{aligned}
\end{equation}
Here, $\boldsymbol{\rho}^{(\ell)}_{k,j}$ describes how the total matching score is distributed over different utility directions, and $\nu^{(\ell)}_j$ describes the normalized scale of the candidate representation, where $m_\ell$ and $d$ denote its number of rows and embedding dimension, respectively.

The scorer is implemented as a lightweight MLP \(f_{\psi}(\cdot)\).
For each target client, we normalize the direct scores over candidate clients and optimize the scorer with mean squared error:
\begin{equation}
\mathcal{L}_{score}
=
\sum_{\ell}
\sum_{k=1}^{K}
\sum_{j=1}^{K}
\left(
f_{\psi}
\left(
\mathbf{x}^{(\ell)}_{k,j}
\right)
-
\operatorname{Norm}_{j}
\left(
s^{(\ell)}_{k,j}
\right)
\right)^2 .
\label{eq:score_loss}
\end{equation}

The final retrieval score is obtained by combining the direct score and the refined score:
\begin{equation}
r^{(\ell)}_{k,j}
=
s^{(\ell)}_{k,j}
+
\lambda
\operatorname{Std}_{q}
\left(
s^{(\ell)}_{k,q}
\right)
\operatorname{Norm}_{j}
\left(
f_{\psi}
\left(
\mathbf{x}^{(\ell)}_{k,j}
\right)
\right),
\end{equation}
where $\lambda$ controls the contribution of score refinement.
The same scorer $f_{\psi}(\cdot)$ is shared across all levels and item groups, since it learns a general rule for refining retrieval scores rather than a level-specific relation.
Finally, the refined scores are normalized across candidate clients, and the aggregation weights are computed by applying softmax over $r^{(\ell)}_{k,j}$.

\subsection{Personalized Filter Aggregation}

After obtaining the refined retrieval scores, the server aggregates candidate filters to construct a personalized item-item filter for each target client.
For each level $\ell$, the retrieval score $r^{(\ell)}_{k,j}$ is normalized over candidate clients by softmax:
\begin{equation}
\alpha^{(\ell)}_{k,j}
=
\frac{
\exp
\left(
r^{(\ell)}_{k,j}/\tau
\right)
}{
\sum_{q=1}^{K}
\exp
\left(
r^{(\ell)}_{k,q}/\tau
\right)
},
\end{equation}
where $\tau$ is the temperature parameter.
Then, the level-$\ell$ personalized residual filter is computed as:
\begin{equation}
\mathbf{R}^{(\ell)}_{k}
=
\sum_{j=1}^{K}
\alpha^{(\ell)}_{k,j}
(\mathbf{F}^{(\ell)}_j - \mathbf{F}^{(\ell)}_g).
\end{equation}
Here, $\mathbf{R}^{(\ell)}_{k}$ is defined at the granularity of level $\ell$.
Before combining residual filters from different levels, the server maps each $\mathbf{R}^{(\ell)}_{k}$ back to the full item-item filter space.

The final personalized item-item filter of client $k$ is:
\begin{equation}
\tilde{\mathbf{S}}_k
=
\mathbf{S}_g
+
\sum_{\ell=1}^{L}
\beta_\ell
\operatorname{Map}_{\ell}
\left(
\mathbf{R}^{(\ell)}_{k}
\right),
\label{eq:personalized_filter}
\end{equation}
where $\operatorname{Map}_{\ell}(\cdot)$ maps the level-$\ell$ filter to the full item-item filter space by assigning the item relation by the value of its corresponding group relation, and $\beta_\ell$ controls the contribution of the $\ell$-th level.
In this way, the personalized filter preserves the global collaborative information from $\mathbf{S}_g$ and incorporates utility-guided residual information from different levels.
In implementation, we optionally use parameter-based aggregation at the coarse level on large-scale datasets to reduce computation cost.
This acceleration is only applied to coarse-level relations and does not change the utility-guided residual aggregation at finer levels.

\subsection{Local Fine-Tuning with Personalized Filter}

After obtaining the personalized filter $\tilde{\mathbf{S}}_k$, the server sends it to client $k$.
Following collaborative-information-based federated recommendation, client $k$ locally fine-tunes its item embeddings by matching their item-item filter with $\tilde{\mathbf{S}}_k$:
\begin{equation}
\min_{\mathbf{E}_k}
\left\|
\mathbf{E}_k\mathbf{E}_k^\top
-
\tilde{\mathbf{S}}_k
\right\|_F^2 .
\label{eq:local_finetune}
\end{equation}
Only item-side embeddings are updated in this step, while user-side parameters remain local and are not uploaded.

\subsection{Algorithm}

Algorithm~\ref{alg:ours} summarizes one communication round of \ours.

\begin{algorithm}[t]
\caption{One Communication Round of \ours}
\label{alg:ours}
\begin{algorithmic}[1]
\REQUIRE Client datasets $\{\mathcal{D}_k\}_{k=1}^{K}$, number of levels $L$, projection matrices $\{\mathbf{P}^{(\ell)}\}_{\ell=1}^{L}$
\ENSURE Personalized filters $\{\tilde{\mathbf{S}}_k\}_{k=1}^{K}$

\FOR{each client $k=1,\dots,K$ in parallel}
    \STATE Train the local recommendation model by Eq.~\eqref{eq:local_training}.
    \STATE Upload item embeddings $\mathbf{E}_k$ to the server.
\ENDFOR

\STATE The server computes $\mathbf{E}_g$ and $\mathbf{S}_g$ by Eq.~\eqref{eq:global_information}.
\STATE The server constructs coarse-to-fine item groups from $\mathbf{E}_g$ and computes multi-level filters $\{\mathbf{F}^{(\ell)}_{k,a}, \mathbf{F}^{(\ell)}_{g,a}\}$.

\FOR{each client $k=1,\dots,K$ in parallel}
    \FOR{each constructed filter $(\ell,a)$}
        \STATE Compute local filter $\mathbf{H}^{(\ell)}_{k,a}$ and utility query $\mathbf{U}^{(\ell)}_{k,a}$ by Eq.~\eqref{eq:utility_query}.
        \STATE Upload projected query $\mathbf{Q}^{(\ell)}_{k,a}$ by Eq.~\eqref{eq:projected_query}.
    \ENDFOR
\ENDFOR

\FOR{each target client $k=1,\dots,K$}
    \FOR{each constructed filter $(\ell,a)$}
        \STATE Compute retrieval scores by Eq.~\eqref{eq:projected_retrieval_score}.
        \STATE Refine retrieval scores with the scorer from the previous round and compute aggregation weights.
        \STATE Aggregate candidate residual filters to obtain $\mathbf{R}^{(\ell)}_{k,a}$.
    \ENDFOR
\ENDFOR
\STATE Update the shared scorer $f_{\psi}(\cdot)$ using Eq.~\eqref{eq:score_loss}.

\FOR{each client $k=1,\dots,K$ in parallel}
    \STATE The server constructs $\tilde{\mathbf{S}}_k$ by Eq.~\eqref{eq:personalized_filter} and sends it to client $k$.
    \STATE Client $k$ fine-tunes item embeddings by Eq.~\eqref{eq:local_finetune}.
\ENDFOR

\end{algorithmic}
\end{algorithm}

\subsection{Privacy Protection}

\ours follows the standard privacy setting in federated recommendation.
Similar to FedMF and other federated recommendation methods, clients only upload item information for collaborative aggregation, and existing privacy strategies such as secure aggregation or differential privacy can be applied to these uploaded parameters when needed.

In addition, \ours requires clients to upload utility signals for client retrieval, which are computed from local filters.
Similar to FedCIA, they are derived from item filter information, so \ours does not introduce additional privacy assumptions beyond existing federated recommendation methods.
In summary, \ours focuses on hierarchical utility-guided relation learning and adopts the same privacy protection strategy as existing federated recommendation methods.

\subsection{Communication Cost Discussion}

\ours follows the communication paradigm of FedCIA and communicates item information for collaborative aggregation.
Compared with FedCIA, \ours introduces extra communication for hierarchical relation learning, including hierarchical filters and projected utility queries.
However, these additional messages are lightweight.
The hierarchical filters are constructed at different granularities and only contain group-level item-item information, rather than multiple full item-item filters.
The utility queries are further compressed by random projection before uploading.
Therefore, the extra communication cost is smaller than communicating one additional full item-item filter, and the total communication cost of \ours is less than twice that of FedCIA in the uncompressed setting.

Moreover, \ours can directly adopt the compression strategy used in FedCIA, such as SVD-based low-rank compression.
Thus, \ours introduces limited additional communication cost while enabling hierarchical utility-guided relation learning.

\begin{table}[t]
\centering
\caption{The statistics of our datasets.
}
\resizebox{0.8\columnwidth}{!}
{
\begin{tabular}{l|ccc}
\toprule
Dataset & \# Users & \# Items & \# Interactions \\
\midrule
ML-100K  & 943      & 1682     & 100000          \\
ML-1M    & 6040     & 3706     & 1000209         \\
Book    &11000    &9332      &200860             \\
BX       &12794      &13775    &146101            \\
Beauty      &5083    &11479      &90952            \\
\bottomrule
\end{tabular}
}
\label{tab:dataset}
\end{table}

\begin{table*}[t]
\centering
\caption{Overall performance on five datasets. The bold indicates the best and the underline indicates the second-best. }
\label{tab:overall}
\resizebox{0.9\textwidth}{!}{
\begin{tabular}{llccccccccc}
\toprule
Dataset & Metric & MF & FedMF & FedNCF & PFedRec & GPFedRec & FedRAP & FedCIA & FedCA & \textbf{\ours} \\
\midrule
\multirow{3}{*}{ML-100K}
& Recall@10 & 0.0756 & 0.1813 & 0.1680 & 0.1681 & 0.1162 & 0.0600 & \underline{0.1903} & 0.1184 & \textbf{0.1973} \\
& MRR@10  & 0.3232 & 0.5501 & 0.5568 & 0.5685 & 0.4550 & 0.0600 & \underline{0.5855} & 0.4107 & \textbf{0.6146} \\
& NDCG@10 & 0.1504 & 0.3198 & 0.3170 & 0.3233 & 0.2395 & 0.1260 & \underline{0.3440} & 0.2008 & \textbf{0.3693} \\
\midrule
\multirow{3}{*}{ML-1M}
& Recall@10 & 0.0641 & 0.1285 & 0.1298 & 0.1177 & 0.0798 & 0.0592 & \underline{0.1392} & 0.0806 & \textbf{0.1407} \\
& MRR@10  & 0.3601 & 0.5664 & 0.5590 & 0.5367 & 0.4269 & 0.3082 & \underline{0.5947} & 0.4167 & \textbf{0.6038} \\
& NDCG@10 & 0.1816 & 0.3254 & 0.3265 & 0.3192 & 0.2396 & 0.1572 & \underline{0.3611} & 0.1989 & \textbf{0.3640} \\
\midrule
\multirow{3}{*}{Book}
& Recall@10 & 0.0412 & 0.1064 & \underline{0.1067} & 0.0300 & 0.0464 & 0.0281 & 0.1012 & 0.0381 & \textbf{0.1075} \\
& MRR@10  & 0.0643 & 0.1265 & \underline{0.1333} & 0.0414 & 0.0667 & 0.0469 & 0.1245 & 0.0569 & \textbf{0.1350} \\
& NDCG@10 & 0.0371 & \underline{0.0846} & \textbf{0.0875} & 0.0239 & 0.0396 & 0.0259 & 0.0816 & 0.0335 & \textbf{0.0875} \\
\midrule
\multirow{3}{*}{BX}
& Recall@10 & 0.0068 & 0.0243 & 0.0211 & 0.0094 & 0.0231 & 0.0041 & \underline{0.0263} & 0.0060 & \textbf{0.0284} \\
& MRR@10  & 0.0070 & 0.0182 & 0.0205 & 0.0089 & 0.0205 & 0.0041 & \underline{0.0257} & 0.0058 & \textbf{0.0288} \\
& NDCG@10 & 0.0047 & 0.0150 & 0.0153 & 0.0062 & 0.0155 & 0.0029 & \underline{0.0186} & 0.0039 & \textbf{0.0208} \\
\midrule
\multirow{3}{*}{Beauty}
& Recall@10 & 0.0484 & 0.0514 & 0.0616 & 0.0401 & 0.0414 & 0.0343 & \underline{0.0660} & 0.0466 & \textbf{0.0711} \\
& MRR@10  & 0.0783 & 0.0632 & 0.0791 & 0.0515 & 0.0642 & 0.0587 & \underline{0.0815} & 0.0686 & \textbf{0.0886} \\
& NDCG@10 & 0.0447 & 0.0418 & 0.0506 & 0.0334 & 0.0375 & 0.0334 & \underline{0.0541} & 0.0407 & \textbf{0.0581} \\
\bottomrule
\end{tabular}
}
\end{table*}

\section{Experiments}

We aim to answer the following research questions:

\textbf{RQ1:} Does \ours outperform existing federated recommendation methods?

\textbf{RQ2:} Are our hierarchical utility-guided relations better than existing relations?

\textbf{RQ3:} Do the utility query and learned scorer effectively estimate useful candidate clients?

\textbf{RQ4:} Does \ours introduce acceptable communication and computation cost?

\subsection{Experimental Settings}

\subsubsection{Datasets}

We evaluate \ours on five real-world recommendation datasets, including ML-100K~\cite{harper2015movielens}, ML-1M~\cite{harper2015movielens}, BX~\cite{ziegler2005improving}, Book~\cite{mcauley2015image}, and Beauty~\cite{mcauley2015image}. The statistics of these datasets are shown in Table~\ref{tab:dataset}.
Following the same setting as previous federated recommendation studies, we split each dataset into training, validation, and test sets, and divide users into 100 clients. 

Our method's computations are \textbf{based solely on the training set}, eliminating unfair practices such as exploiting validation errors.

\subsubsection{Baselines}
We compare \ours with representative federated recommendation methods, including FedMF~\cite{chai2020secure}, FedNCF~\cite{perifanis2022federated}, PFedRec~\cite{zhang2023dual}, GPFedRec~\cite{zhang2024gpfedrec}, FedRAP~\cite{li2024federated}, FedCIA~\cite{han2025fedcia}, and FedCA~\cite{zhang2026beyond}.

\subsubsection{Evaluation Metrics}

We adopt Recall~\cite{herlocker2004evaluating}, Mean Reciprocal Rank (MRR)~\cite{voorhees1999trec}, and Normalized Discounted Cumulative Gain (NDCG)~\cite{jarvelin2002cumulated} as evaluation metrics.
These metrics are widely used in federated recommendation to evaluate performance.
Following common practice, we report Recall@10, MRR@10, and NDCG@10.

\subsubsection{Implementation Details}

We implement \ours using PyTorch on an NVIDIA Tesla T4 GPU. We use MF as the backbone model for \ours. 
For a fair comparison, all baselines and \ours follow the same hyperparameter search space.
Specifically, the learning rate is searched in $\{0.1, 0.01, 0.001\}$, and the weight decay is searched in $\{0,10^{-6},10^{-3}\}$, and the batch size is searched in $\{256, 1024, 4096\}$. We use Adam as the optimizer; the details are provided in the reproduction code.

It is worth noting that some personalized federated learning methods evaluate recommendation performance with a sampled candidate set, which usually contains one positive item and 99 negative items. 
This evaluation is different from the full-ranking setting and may lead to different performance scales. 
In our experiments, we follow the standard recommendation evaluation~\cite{he2020lightgcn} and rank each positive item against all items. 
Therefore, the results reported in this paper are not directly comparable with those obtained under sampled evaluation settings.

\subsection{Overall Performance (RQ1)}

Table~\ref{tab:overall} reports the overall recommendation performance on five datasets.
From the results, we have the following observations.
1) Federated recommendation methods generally outperform the basic recommendation model without aggregation.
This shows that aggregating information from multiple clients is useful in our experimental setting.
2) Personalized aggregation methods based on predefined relations do not always achieve better performance.
This indicates that client relations constructed from predefined similarity or complementarity do not necessarily correspond to useful information for prediction.
Therefore, a client selected by a predefined relation may not provide the information that the target client actually needs.
3) \ours consistently achieves the best or comparable performance across all datasets.
Compared with baselines, \ours further improves the results on all datasets.
This demonstrates that our hierarchical utility-guided relations can further help each target client retrieve more useful candidate information.

\begin{table}[t]
\centering
\caption{Relation analysis under the same aggregation framework.}
\label{tab:ablation}
\resizebox{\columnwidth}{!}{
\begin{tabular}{lccc}
\toprule
Method & Recall@10 & MRR@10 & NDCG@10 \\
\midrule
No Relation (FedCIA) & 0.1903 & 0.5855 & 0.3440 \\
\midrule
Single Similarity & 0.1904 & 0.5948 & 0.3442 \\
Single Complementarity & 0.1933 & 0.6045 & 0.3567 \\
Single S + C & 0.1930 & 0.5902 & 0.3473 \\
Single Utility & 0.1954 & 0.6047 & 0.3549 \\
\midrule
Hierarchical Similarity & \underline{0.1962} & 0.6064 & 0.3582 \\
Hierarchical Complementarity & 0.1950 & \underline{0.6108} & \underline{0.3657} \\
Hierarchical S + C & 0.1948 & 0.6074 & 0.3619 \\
Hierarchical Utility (ours) & \textbf{0.1973} & \textbf{0.6146} & \textbf{0.3693} \\
\bottomrule
\end{tabular}
}
\end{table}

\begin{table}[t]
\centering
\caption{Ablation study of utility query construction.}
\label{tab:ablation_utility}
\resizebox{0.9\columnwidth}{!}{
\begin{tabular}{lccc}
\toprule
Method & Recall@10 & MRR@10 & NDCG@10 \\
\midrule
$U=\mathbf{H}$ 
& 0.1932 & 0.6074 & 0.3507 \\
$U=\mathbf{H}-\mathbf{H}\bar{\mathbf{F}}$ 
& 0.1956 & 0.6081 & 0.3626 \\
$U=\mathbf{H}-\bar{\mathbf{F}}\mathbf{H}$ 
& 0.1919 & 0.5846 & 0.3471 \\
$U=\mathbf{H}-\frac{1}{2}(\mathbf{H}\bar{\mathbf{F}}+\bar{\mathbf{F}}\mathbf{H})$ 
& \textbf{0.1973} & \textbf{0.6146} & \textbf{0.3693} \\
\bottomrule
\end{tabular}
}
\end{table}

\subsection{Relation Analysis (RQ2)}

Although Table~\ref{tab:overall} shows that \ours is better than existing methods, these methods may use different backbone models, communicated information, and aggregation objects.
Therefore, the comparison in Table~\ref{tab:overall} cannot fully indicate whether our relation is better.
Therefore, in this section, we fix the backbone, communicated information, and aggregation object to make a fair comparison.
Specifically, we take FedCIA as the baseline and extract the following client relations based on existing methods.
1) No Relation. Some traditional methods~\cite{chai2020secure,zhang2023dual,han2025fedcia} simply average all uploaded information without designing specific relations based on the target client.
2) Similarity Relation. Some methods~\cite{zhang2024gpfedrec} define the relations between clients by the similarity of the parameters of the uploaded models.
3) Complementarity Relation. In addition to similarity, some methods~\cite{zhang2026beyond} suggest the complementarity among models should also be considered.
4) Similarity and Complementarity (S + C) Relation. The combination of the above relations.
For each relation type, we compare both single and hierarchical relations to evaluate whether hierarchical relation modeling is helpful.
Table~\ref{tab:ablation} reports the relation analysis on ML-100K.

From the results, we have the following observations.
1) Compared with existing similarity-based, complementarity-based, similarity and complementarity-based relations, our utility-guided  relation achieves better performance.
This shows that those predefined relations do not always indicate whether a candidate client can improve the target client's performance after aggregation.
In contrast, our utility-guided relation directly estimates whether the candidate information matches the target client's required collaborative direction, making it more suitable for personalized federated recommendation.
2) For all relation types, the hierarchical relations generally outperform the single relation.
This demonstrates that a single global relation is insufficient to describe complex user relations in recommendation.
Therefore, the Hierarchical Utility relation achieves the best overall performance, verifying the effectiveness of the proposed hierarchical utility-guided client relation.

\subsection{Ablation Study (RQ3)}

\subsubsection{Effect of utility query construction.}

We further analyze different ways of constructing the utility query.
As shown in Table~\ref{tab:ablation_utility}, directly using the local filter $\mathbf{H}$ can already provide a useful signal, but its performance is clearly worse than the full symmetric form.
This indicates that only considering local signal is insufficient, because the utility query should also consider what information has already been captured by the current base filter.
The two residual queries, i.e., $\mathbf{H}-\mathbf{H}\bar{\mathbf{F}}$ and $\mathbf{H}-\bar{\mathbf{F}}\mathbf{H}$, achieve unstable performance across different metrics.
In contrast, the symmetric form achieves the best results on all metrics.
This verifies the effectiveness of our utility query design and supports the use of the utility direction in \ours.

\subsubsection{Effect of scorer input features.}
\begin{table}[t]
\centering
\caption{Ablation study of scorer input features.}
\label{tab:ablation_scorer}
\resizebox{0.9\columnwidth}{!}{
\begin{tabular}{lccc}
\toprule
Method & Recall@10 & MRR@10 & NDCG@10 \\
\midrule
score\_only $(\sqrt{s})$ 
& 0.1960 & 0.6062 & 0.3571 \\
rho\_only $(\boldsymbol{\rho})$ 
& 0.1962 & 0.5935 & 0.3522 \\
no\_rho $(\nu+\sqrt{s})$ 
& 0.1934 & 0.6065 & 0.3523 \\
no\_nu $(\boldsymbol{\rho}+\sqrt{s})$ 
& 0.1954 & 0.6062 & 0.3599 \\
no\_score $(\boldsymbol{\rho}+\nu)$ 
& 0.1967 & 0.6095 & 0.3621 \\
full $(\boldsymbol{\rho}+\nu+\sqrt{s})$ 
& \textbf{0.1973} & \textbf{0.6146} & \textbf{0.3693} \\
\bottomrule
\end{tabular}
}
\end{table}

\begin{table}[t]
\centering
\small
\setlength{\tabcolsep}{3.5pt}
\caption{Correlation between pre-aggregation relation scores and post-aggregation performance on ML-100K. P and S denote Pearson and Spearman correlations, respectively.}
\label{tab:correlation}
\resizebox{0.9\columnwidth}{!}{
\begin{tabular}{llrrrr}
\toprule
Split & Metric & Sim-P & Util-P & Sim-S & Util-S \\
\midrule
Train & Recall@10 & -0.1125 & \textbf{0.1073} & -0.1835 & \textbf{0.1632} \\
Train & MRR@10    & \textbf{0.1087} & -0.1106 & \textbf{0.0871} & -0.1002 \\
Train & NDCG@10   & -0.0517 & \textbf{0.0332} & -0.1333 & \textbf{0.0950} \\
\midrule
Valid & Recall@10 & -0.2606 & \textbf{0.2634} & -0.1738 & \textbf{0.1772} \\
Valid & MRR@10    & -0.3300 & \textbf{0.3327} & -0.2886 & \textbf{0.3030} \\
Valid & NDCG@10   & -0.2797 & \textbf{0.2802} & -0.2761 & \textbf{0.2775} \\
\midrule
Test  & Recall@10 & -0.2520 & \textbf{0.2553} & -0.3219 & \textbf{0.3127} \\
Test  & MRR@10    & -0.2235 & \textbf{0.2232} & -0.2220 & \textbf{0.2156} \\
Test  & NDCG@10   & -0.3248 & \textbf{0.3224} & -0.3469 & \textbf{0.3376} \\
\bottomrule
\end{tabular}
}
\end{table}

Table~\ref{tab:ablation_scorer} reports the ablation results of different scorer inputs.
The score\_only performs much worse than the full model, showing that the learned scorer is not simply a score mapping. Only using the total matching score cannot capture the matching score across different projected directions.
Compared with score\_only, using $\boldsymbol{\rho}$ or $\nu$ generally achieves better performance, indicating that both direction matching information and embedding scale statistics are useful for score refinement.
The no\_rho and no\_nu further show that removing either component weakens performance and no\_score still performs competitively, showing that the decomposed matching features also contain strong utility information.
Finally, the full scorer achieves the best results on all metrics, demonstrating that combining direction matching strength, embedding norm statistics, and total matching score provides the most stable retrieval performance.

\subsubsection{Correlation analysis.}
To examine whether relation scores reflect actual aggregation
benefits, we compute Pearson and Spearman correlations between
pre-aggregation relation scores and post-aggregation performance
on ML-100K. As shown in
Table~\ref{tab:correlation}, utility scores are consistently positively
correlated with validation and test performance, while similarity
scores are unstable and mostly negative. This suggests that
similarity captures general relatedness, whereas the proposed
utility score is better aligned with the target client's current
aggregation need.

\subsubsection{Visualization of Hierarchical Relations}

\begin{figure}[t]
\centering
\includegraphics[width=0.95\columnwidth]{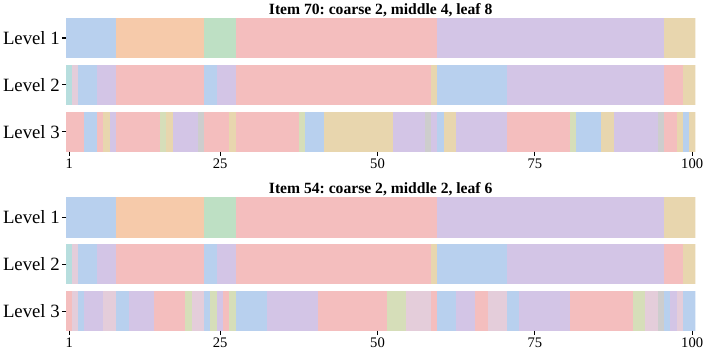}
\caption{
Visualization of hierarchical relations for two randomly selected items.
Each row represents one level, and different colors indicate the item relations used at this level.
}
\label{fig:multilayer}
\end{figure}

To further illustrate whether \ours truly constructs hierarchical relations, we visualize the relations used by two randomly selected items in Figure~\ref{fig:multilayer}.
Each row corresponds to one relation level, and the highlighted regions indicate the relations at that level.

From the figure, we can observe that the same target item relies on different relations across different levels.
At the coarse level, the item is related to a broad item group, which captures general collaborative patterns shared by many items.
At finer levels, the relations become more localized, where the item only uses relations within fine groups.
Therefore, the relations used by the same item are gradually refined from coarse to fine.
This visualization verifies that \ours assigns hierarchical relations at different granularities.

\begin{figure}[t]
\centering
\includegraphics[width=\columnwidth]{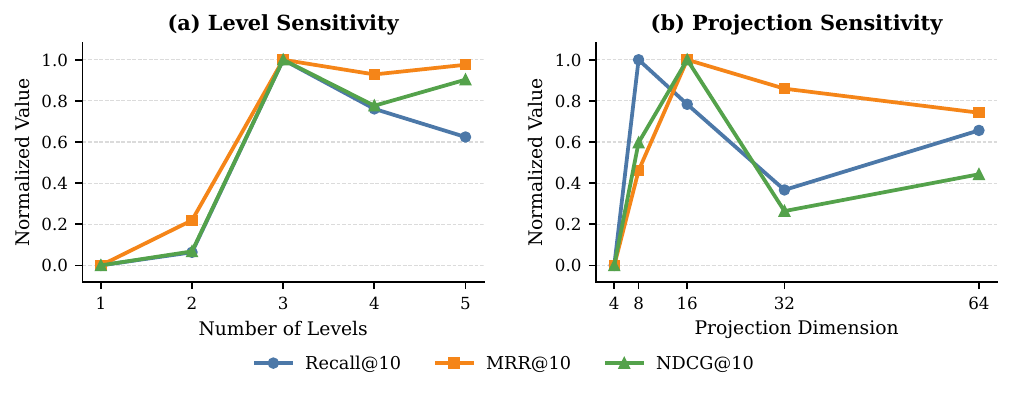}
\caption{Sensitivity analysis of number of levels and projection dimension. 
For clarity, each metric is independently normalized to $[0,1]$.}
\label{fig:sensitivity}
\end{figure}

\subsubsection{Sensitivity Analysis.}
We analyze the sensitivity of \ours to the number of levels and the projection dimension in Figure~\ref{fig:sensitivity}. \ours performs best with three levels, showing that moderate coarse-to-fine relation modeling is beneficial, while overly deep hierarchies may introduce noisy or redundant relations.
For the projection dimension, performance first improves and then becomes stable or slightly decreases, indicating that a moderate projection size is sufficient to preserve utility directions while reducing communication cost.

\subsection{Efficiency Analysis (RQ4)}

\begin{table}[t]
\centering
\caption{Communication and runtime comparison with FedCIA.}
\label{tab:efficiency}
\resizebox{\columnwidth}{!}{
\begin{tabular}{lccccc}
\toprule
Method 
& Upload 
& Download 
& Runtime 
& Finetune Time 
& Total Time \\
\midrule
FedCIA 
& 0.82 MB 
& 0.82 MB 
& 20.10 s 
& 11.68 s 
& 610.04 s \\
\ours 
& 0.92 MB 
& 1.16 MB 
& 29.06 s 
& 11.66 s 
& 700.05 s \\
\bottomrule
\end{tabular}
}
\end{table}

Table~\ref{tab:efficiency} reports the communication and runtime comparison between FedCIA and \ours.
Since \ours is built upon the collaborative information aggregation paradigm of FedCIA, it introduces extra communication for hierarchical filters and projected utility queries.
However, the additional communication cost is limited.
The total communication cost of \ours is about 27\% higher than FedCIA, which is far below twice the communication cost.
In terms of runtime, \ours requires additional computation for hierarchical relation mining, utility-guided retrieval, and score refinement.
Thus, the total training time of \ours is 700.05 s, only about 14.8\% higher than FedCIA.
These results show that \ours introduces limited additional communication and computation cost while achieving better recommendation performance.

\section{Conclusion}

In this paper, we propose \ours, a federated recommendation framework that learns \textbf{hierarchical utility-guided client relations} for personalized aggregation.
\ours contains two key modules.
The hierarchical relation mining module constructs multi-level relations, enabling relation modeling from coarse to fine.
The utility-guided client retrieval module derives utility signals from each target client's local information and retrieves candidate clients that match the target client's utility direction.
Extensive experiments on five real-world datasets demonstrate that \ours consistently improves recommendation performance over representative federated recommendation baselines.

\begin{acks}

This work is supported by Major Project of the National Social Science Fund of China (NSSFC) under the Grant No. 25\&ZD260, and National Natural Science Foundation of China (NSFC) under the Grant No. 62372113. Peng Zhang \& Tun Lu are with the College of Computer Science and Artificial Intelligence, Fudan University. Tun Lu is also affiliated to the MOE Laboratory for National Development and Intelligent Governance, and Shanghai Key Laboratory of Data Science, Fudan University.

\end{acks}

\section*{GenAI Usage Disclosure}

The authors used generative AI tools only for language polishing and grammar checking during the preparation of this manuscript. 
All technical ideas, method design, experimental results, analyses, and conclusions were produced and verified by the authors. 
The authors carefully reviewed and revised all AI-assisted text and take full responsibility for the final content of the paper.

\bibliographystyle{ACM-Reference-Format}
\balance
\bibliography{ref}

\end{document}